\documentclass[submission,copyright,creativecommons]{eptcs}
\providecommand{\event}{MCU 2013} 

\title{Tiling problems on Baumslag-Solitar groups.}
\author{Nathalie Aubrun
\institute{LIP, ENS de Lyon, CNRS, INRIA, UCBL, Universit\'e de Lyon }
\email{nathalie.aubrun@ens-lyon.fr}
\and
Jarkko Kari
\institute{Department of Mathematics and Statistics, University of Turku}
\email{jkari@utu.fi}
}
\def\titlerunning{Tiling problems on Baumslag-Solitar groups.}
\def\authorrunning{N. Aubrun \& J. Kari}

\usepackage{latexsym,amsfonts,amsmath,amssymb,amsthm,amscd,stmaryrd,mathrsfs}

\usepackage{bbm}
\usepackage{yhmath}

\usepackage[dvips]{epsfig}
\usepackage{graphicx,graphics,psfrag}
\usepackage{float}
\usepackage{color}
\usepackage[table]{xcolor}
\usepackage{subfigure}

\usepackage{tikz}
\usetikzlibrary{shapes,patterns,positioning,arrows,decorations.pathmorphing}

\makeatletter

\newskip\@bigflushglue \@bigflushglue = -100pt plus 1fil

\def\bigcentering{\let\\\@centercr\rightskip\@bigflushglue%
\leftskip\@bigflushglue
\parindent\z@\parfillskip\z@skip}

\makeatother

\newcommand{\G}{G}
\newcommand{\R}{\mathbbm{R}}
\newcommand{\Z}{\mathbbm{Z}}
\newcommand{\N}{\mathbbm{N}}
\renewcommand{\L}{\mathcal{L}}

\newcommand{\Gen}{\mathcal{G}}
\newcommand{\Rel}{\mathcal{R}}
\newcommand{\Per}{Per}

\newcommand{\BS}[2]{\mathbf{BS}(#1,#2)}
\newcommand{\floor}[1]{\lfloor #1\rfloor}
\newcommand{\bigfloor}[1]{\left\lfloor #1\right\rfloor}
\newcommand{\define}{\emph}

\theoremstyle{plain}
\newtheorem{theorem}{Theorem}
\newtheorem{proposition}[theorem]{Proposition}

\newtheorem{lemma}[theorem]{Lemma}

\theoremstyle{definition}

\theoremstyle{remark}

\newtheorem{example}{Example}[section]

\begin{document}
\maketitle

\begin{abstract}
We exhibit a weakly aperiodic tile set for Baumslag-Solitar groups, and prove that the domino problem is undecidable on these groups. A consequence of our construction is the existence of an arecursive tile set on Baumslag-Solitar groups.
\end{abstract}

\section*{Introduction}

Wang tilings are colorings of the Euclidean plane that respect some local constraints. They can be viewed both as dynamical systems and computational models, but they were first introduced by Hao Wang to study decision problems on some classes of logical formulas~\cite{Wang1961}.

The concept of Wang tiles may be generalized to define tilings on other regular structures than $\Z^2$, for instance on the Cayley graph of a finitely generated group. The domino problem -- deciding whether a finite tile set can tile the group or not -- is undecidable on $\Z^2$~\cite{Berger1966}. A more recent result~\cite{MullerSchupp1985,KuskeLohrey2005} characterizes groups having a decidable monadic second order logic. A consequence of this result states that if the group is virtually free, then it has a decidable domino problem. The reciprocal statement is still an open question, and very few examples of effective constructions to prove the (un-)decidability of the domino problem on finitely presented groups are known~\cite{DBLP:conf/mcu/Kari07,Margenstern2008}. Obviously the domino problem cannot be decidable on groups with undecidable word problem, so it is relevant to investigate non virtually free groups with decidable word problem. A famous theorem by Magnus~\cite{Magnus1932} states that provided a 
finitely presented group possesses a one-relator presentation, it has decidable word problem.

In this article we investigate tilings on Baumslag-Solitar groups, that are two generators and one-relator non virtually free groups. This paper is organized as follows. Section~\ref{section.tilings} is devoted to the definition of tilings on finitely presented groups, which relies on the Cayley graph representation of such a group. In Section~\ref{section.BS_groups} we present the class of Baumslag-Solitar groups, and define Wang tiles on these groups. Following the ideas in~\cite{DBLP:journals/dm/Kari96} and \cite{DBLP:conf/mcu/Kari07}, we exhibit a weakly aperiodic tile set (see Section~\ref{section.aperiodic}), and we prove that the domino problem is undecidable on these groups (see
Section~\ref{section.dp_undecidable}).
Analogously to~\cite{Jeandel2012}, we conclude that there exist arecursive tile sets on the Baumslag-Solitar groups, that is, tile sets that admit valid tilings but all valid tilings are non-recursive.

\section{Tilings on finitely presented groups}
\label{section.tilings}

This section is devoted to definitions necessary to formalize a generalization of Wang tiles to a particular class of regular graphs, that is the class of Cayley graphs of finitely presented groups.

\subsection{Finitely presented groups and graph representation}

The group $\G$ is \define{finitely generated} if it has a presentation $P=<\Gen|\Rel>$ where $\Gen$ is finite, and \define{finitely presented} if it has a presentation $P=<\Gen|\Rel>$ where both $\Gen$ and $\Rel$ are finite (see~\cite{LyndonSchupp1977} for an introduction to group presentations). By abuse of notation, we will write $G=<\Gen|\Rel>$ for saying that the group $G$ has presentation $<\Gen|\Rel>$. The set $\Gen$ is the set of \define{generators} of the group, and the set $\Rel$ is the set of \define{relators} of the group and is made of words on the alphabet $\Gen\cup\Gen^{-1}$ that represent the identity of the group. We adopt here the conventions that if the set of generators contains an element $g\in\G$, then it does not contain its inverse $g^{-1}$, and that the set of relators $\Rel$ does not contains trivial relators $xx^{-1}=\varepsilon$, even if they always hold. We also denote by $\Gen^{-1}$ the set of inverses of generators, and by $\varepsilon$ the identity of $\G$. A group has 
infinitely many presentations, and determining whether two presentations define two isomorphic groups is an undecidable problem~\cite{Rabin1958}.

The \define{Cayley graph} of a finitely presented group $\G$, denoted by $\Gamma_\G=(V_\Gamma,E_\Gamma)$, is a graph whose set of vertices $V_\Gamma$ is identified with $\G$ and whose set of edges $E_\Gamma$ consists of pairs of the form $(h,hg)$, where $h$ is an element of $\G$ and $g$ is a generator or an inverse of generator. Note that the Cayley graph depends on the choice of the generators.

\subsection{Tilings and periodicity}
\label{subsection.tilings_periodicity}

Let $G$ be a group, and let $\mathcal{A}\subseteq G$ be a finite set with $n$ elements. Given a finite set of colors $\mathcal{C}$, a \define{Wang tile on $\mathcal{A}$} is mapping $t:\mathcal{A}\times\{+,-\}\rightarrow\mathcal{C}$, where $(g,+)$ is the color on the outgoing edge $g$ and $(g,-)$ is the color on the incoming edge $g$. A Wang tile can be represented by a polygon with $2n$ edges, and every edge is associated with a color. Let $\tau$ be a tile set on $G$. A \define{valid tiling of $G$ by $\tau$} is a map $t:G\rightarrow\tau$ that satisfies the following condition
$$\forall x\in G,\forall g\in\mathcal{A}\ \ t(x)(g,+)=t(xg)(g,-),$$ 
that is to say that two adjacent Wang tiles have the same color on their common edge. A tiling can thus be seen
as a coloring of the edges of the Cayley graph of $G$ that satisfies given local constraints at each vertex.

Let $\tau$ be a finite tile set. Then the group $\G$ acts on the set of tilings $\tau^\G$ by translation. Given $m\in\G$, we define $\mathfrak{S}_m$ the \define{translation by $m$} by
$$\mathfrak{S}_m:\left(
\begin{array}{rcl}
\tau^\G & \rightarrow & \tau^\G \\
 x & \mapsto & \mathfrak{S}_m(x)
\end{array}
\right),$$
where $\left[\mathfrak{S}_m(x)\right]_g=x_{mg}$ for any $g\in\G$. Let $x\in\tau^\G$ be a valid tiling of $\G$ by $\tau$. Given an element $m\in\G$, the tiling $x$ is said to be \define{$m$-periodic} if $x=\mathfrak{S}_m(x)$. In that case we say that $m$ is a \define{period} for $x$. The set of periods of a tiling $x\in\tau^\G$, denoted by $\Per(x)$, is thus a sub-group of $\G$.

Following Goodman-Strauss~\cite{GoodmanStrauss2000} we define two notions of aperiodicity. A tiling $x$ is \define{weakly periodic} if $\Per(x)$ contains an infinite cyclic subgroup, and $x$ is \define{strongly non-periodic} if it is not weakly periodic. A tile set is \define{strongly aperiodic} if a valid tiling exists and if it admits only strongly non-periodic tilings. A tiling $x$ is \define{strongly periodic} if $\Per(x)$ is a finite index subgroup of $G$. A tiling $x$ is \define{weakly non-periodic} if it is not strongly periodic. A tile set is \define{weakly aperiodic} if a valid tiling exists and if it admits only weakly non-periodic tilings.

It is clear that if $G$ contains an element of infinite order, the strong concepts imply the corresponding weak concepts. It is also well known that in $G=\Z$ the strong and weak concepts are equivalent, and in $G=\Z^2$ a tile set is strongly aperiodic if and only if it is weakly so. In contrast, already in $G=\Z^3$ there are tile sets that are weakly aperiodic but not strongly aperiodic. In the following we see the same phenomenon on the class of Baumslag-Solitar groups, that will be defined in Section~\ref{subsection.def_BS}.

\subsection{The domino problem}
\label{subsection.dp}

Let $G$ be a finitely generated group, and suppose a finite set of generators $\Gen$ is fixed. Is it possible to find an algorithm that takes as input a finite set of Wang tiles $\tau$ on $\Gen$, and outputs $\mathbf{Yes}$ if and only if there exists a valid tiling by $\tau$ ? This problem is known as the \define{domino problem} on group $G$. The decidability status
of this problem does not depend on the choice of the generator set $\Gen$ by the following theorem.

\begin{theorem}\label{theorem.deci_presentation}
Let $G$ be a group, and let $\mathcal A$ and $\mathcal B$ be two finite subsets of $G$. Suppose that $\mathcal B$ is contained in the subgroup generated by $\mathcal A$. Then, if it is decidable whether a given Wang tile set on $\mathcal A$ admits a tiling, then the analogous question on $\mathcal B$ is also decidable.
\end{theorem}

It follows that the choice of the generators is irrelevant for the decidability of the domino problem of a finitely generated group. It also follows from the previous theorem that if a finitely generated group $G$ is a subgroup of a finitely generated group with decidable domino problem, then the domino problem on $G$ is also decidable.


\begin{theorem}\label{theorem.dp_undeci_Z2}
 The domino problem is undecidable on $\Z^2$.
\end{theorem}

The original proof of this result due to Berger~\cite{Berger1966}, for the usual presentation $\Z^2=<a,b|ab=ba>$, relies on the
construction of a strongly aperiodic tile set. In this paper we adopt the method introduced in~\cite{DBLP:conf/mcu/Kari07} to prove the undecidability of the domino problem on Baumslag-Solitar groups.


\subsection{Tilings defined by allowed patterns}
\label{subsection.sft}

It is well known that tilings as they were introduced in Section~\ref{subsection.tilings_periodicity} are the same objects as
subshifts of finite type (SFT), studied in symbolic dynamics. Let $A$ be a finite alphabet. An SFT on group $G$ is determined by a finite set of allowed patterns $p\in A^{\mathcal D}$ on a finite subset $\mathcal D$ of $G$. A coloring $x\in A^G$ of $G$ is in the corresponding SFT if all its translates have an allowed pattern in domain $\mathcal D$.

One may as well consider colorings of the edges of the Cayley graph of a finitely generated group. In the following discussion we adopt this convenient way to describe tile sets on Baumslag-Solitar groups. By abuse of language we will use the term ``tile'' for a finite allowed pattern, and the term ``tile set'' for a finite set of allowed patterns. It is a simple matter to convert such tile sets into equivalent Wang tile sets.

\section{The Baumslag-Solitar groups}
\label{section.BS_groups}

We now focus on a particular class of groups, the Baumslag-Solitar groups, that were initially introduced in~\cite{BaumslagSolitar1962} as examples of non-Hopfian finitely presented groups. These groups are very simple since they have a presentation with only two generators and one relator, and thus have decidable word problem~\cite{Magnus1932}. Nevertheless their structure is not trivial, as we will see in the sequel.

\subsection{Definition and graph representation}
\label{subsection.def_BS}

Given two non-zero integers $m$ and $n$, we define $\BS{m}{n}$ the \define{Baumslag-Solitar group of order $(m,n)$} as the
two-generators and one-relator group with presentation

$$\BS{m}{n}=<a,b | a^mb=ba^n>.$$

In particular $\BS{1}{1}$ is isomorphic to $\Z^2$. It is known~\cite{Meskin1972} that the group $\BS{m}{n}$ is residually
finite if and only if $|m|=1$ or $|n|=1$ or $|m|=|n|$. Since a finitely presented virtually free group is always residually finite (as a consequence of~\cite{KarrassPietrowskiSolitar1973}), the above property means that most of the Baumslag-Solitar groups do not satisfy the hypothesis of the Muller and Schupp's theorem~\cite{MullerSchupp1985}, hence we do not know anything about the domino problem on these groups.

Since $\BS{-m}{-n}$ is isomorphic to $\BS{m}{n}$, it is enough to consider groups with $m>0$. For simplicity, we also assume that $n>0$. The case $n<0$ is analogous.

We now discuss $\Gamma_{m,n}$, the Cayley graph of $\BS{m}{n}$ for $m,n>0$. Since $\BS{m}{n}$ has two generators, every vertex in its Cayley graph $\Gamma_{m,n}$ has in-degree and out-degree $2$.

The \define{level associated with $g$} of the Cayley graph $\Gamma_{m,n}$ is the induced subgraph obtained by keeping only the
vertices of the coset $g\langle a\rangle = \left\{g.a^k:k\in\Z \right\}$. We denote it by $\L_g$ and we say that the vertex $g$
\define{defines} the level $\L_g$.
The level $L_{gb}$ is a \define{predecessor} of the level $L_g$, while the latter is a \define{successor} of the former, for all group elements $g$. Note that each level has $m$ predecessors and $n$ successors.

Our tilings will be colorings of the edges of the Cayley graph $\Gamma_{m,n}$. The local constraint is given in terms of a set of allowed patterns on the edges $$\varepsilon\longrightarrow a \longrightarrow a^2\longrightarrow\dots \longrightarrow a^m
\longrightarrow a^mb=ba^n$$ and $$\varepsilon\longrightarrow b \longrightarrow ba \longrightarrow ba^2\longrightarrow\dots
\longrightarrow ba^n=a^mb,$$ see the left side of Figure~\ref{figure.basic_pattern_levels} for the case $m=3, n=2$.
For each group element $g$ the pattern of this shape found at position $g$ must be among the allowed patterns.

\begin{figure}[H]
\centering

\begin{tabular}{ll}
\begin{minipage}[l]{0.39\linewidth}
\begin{tikzpicture}[scale=0.4]

\draw[thick] (0,0) -- (0,4);
\draw[thick] (0,4) to [controls=+(45:1) and +(135:1)] (3,4);
\draw[thick] (3,4) to [controls=+(45:1) and +(135:1)] (6,4);
\draw[thick] (6,4) -- (6,0);
\draw[thick] (0,0) to [controls=+(45:0.7) and +(135:0.7)] (2,0);
\draw[thick] (2,0) to [controls=+(45:0.7) and +(135:0.7)] (4,0);
\draw[thick] (4,0) to [controls=+(45:0.7) and +(135:0.7)] (6,0);


\draw (0,4) node[below left]{\small$g.b$};
\draw (3,4) node[below]{\small$g.ba$};
\draw (6,4) node[below right]{\small$g.ba^2$};
\draw (0,0) node[below left]{\small$g$};
\draw (2,0) node[below]{\small$g.a$};
\draw (4,0) node[below]{\small$g.a^2$};
\draw (6,0) node[below]{\small$g.a^3$};

\end{tikzpicture}
\end{minipage}
&
\begin{minipage}[l]{0.6\linewidth}
\begin{tikzpicture}[scale=0.2]

\foreach \x in {0,...,5} {
\draw[thick] (6*\x-1.5,0) -- (6*\x,4);
\draw[thick] (6*\x,4) to [controls=+(45:1) and +(135:1)] (6*\x+3,4);
\draw[thick] (6*\x+3,4) to [controls=+(45:1) and +(135:1)] (6*\x+6,4);
\draw[thick] (6*\x-1.5,0) to [controls=+(45:0.7) and +(135:0.7)] (6*\x+2-1.5,0);
\draw[thick] (6*\x+2-1.5,0) to [controls=+(45:0.7) and +(135:0.7)] (6*\x+4-1.5,0);
\draw[thick] (6*\x+4-1.5,0) to [controls=+(45:0.7) and +(135:0.7)] (6*\x+6-1.5,0);
\draw[thick] (6*\x+6,4) -- (6*\x+6-1.5,0);
}
\draw[thick, dashed] (-3,4) -- (-1,4);
\draw[thick, dashed] (37,4) -- (39,4);
\draw (40,4) node{$\L$};

\foreach \x in {0,...,4} {
\begin{scope}[shift={(3,0)},color=black!50]
\draw[thick] (6*\x+1.5,0-1.5) -- (6*\x,4);
\draw[thick] (6*\x,4) to [controls=+(45:1) and +(135:1)] (6*\x+3,4);
\draw[thick] (6*\x+3,4) to [controls=+(45:1) and +(135:1)] (6*\x+6,4);
\draw[thick] (6*\x+1.5,0-1.5) to [controls=+(45:0.7) and +(135:0.7)] (6*\x+2+1.5,0-1.5);
\draw[thick] (6*\x+2+1.5,0-1.5) to [controls=+(45:0.7) and +(135:0.7)] (6*\x+4+1.5,0-1.5);
\draw[thick] (6*\x+4+1.5,0-1.5) to [controls=+(45:0.7) and +(135:0.7)] (6*\x+6+1.5,0-1.5);
\draw[thick] (6*\x+6,4) -- (6*\x+6+1.5,0-1.5);
\end{scope}
}

\foreach \x in {0,1,2} {
\begin{scope}[scale=3/2,shift={(2,4*2/3)}]
\draw[thick] (6*\x,0) -- (6*\x,4);
\draw[thick] (6*\x,4) to [controls=+(45:1) and +(135:1)] (6*\x+3,4);
\draw[thick] (6*\x+3,4) to [controls=+(45:1) and +(135:1)] (6*\x+6,4);
\draw[thick] (6*\x,0) to [controls=+(45:0.7) and +(135:0.7)] (6*\x+2,0);
\draw[thick] (6*\x+2,0) to [controls=+(45:0.7) and +(135:0.7)] (6*\x+4,0);
\draw[thick] (6*\x+4,0) to [controls=+(45:0.7) and +(135:0.7)] (6*\x+6,0);
\draw[thick] (6*\x+6,4) -- (6*\x+6,0);
\end{scope}
}

\foreach \x in {0,1,2,3} {
\begin{scope}[scale=3/2,shift={(0,4*2/3)},color=black!50]
\draw[thick] (6*\x,0) -- (6*\x-1.5,4-1.5);
\draw[thick] (6*\x-1.5,4-1.5) to [controls=+(45:1) and +(135:1)] (6*\x+3-1.5,4-1.5);
\draw[thick] (6*\x+3-1.5,4-1.5) to [controls=+(45:1) and +(135:1)] (6*\x+6-1.5,4-1.5);
\draw[thick] (6*\x,0) to [controls=+(45:0.7) and +(135:0.7)] (6*\x+2,0);
\draw[thick] (6*\x+2,0) to [controls=+(45:0.7) and +(135:0.7)] (6*\x+4,0);
\draw[thick] (6*\x+4,0) to [controls=+(45:0.7) and +(135:0.7)] (6*\x+6,0);
\draw[thick] (6*\x+6-1.5,4-1.5) -- (6*\x+6,0);
\end{scope}
}

\foreach \x in {0,1,2} {
\begin{scope}[scale=3/2,shift={(4,4*2/3)},color=black!75]
\draw[thick] (6*\x,0) -- (6*\x+1.5,4-2.5);
\draw[thick] (6*\x+1.5,4-2.5) to [controls=+(45:1) and +(135:1)] (6*\x+3+1.5,4-2.5);
\draw[thick] (6*\x+3+1.5,4-2.5) to [controls=+(45:1) and +(135:1)] (6*\x+6+1.5,4-2.5);
\draw[thick] (6*\x,0) to [controls=+(45:0.7) and +(135:0.7)] (6*\x+2,0);
\draw[thick] (6*\x+2,0) to [controls=+(45:0.7) and +(135:0.7)] (6*\x+4,0);
\draw[thick] (6*\x+4,0) to [controls=+(45:0.7) and +(135:0.7)] (6*\x+6,0);
\draw[thick] (6*\x+6+1.5,4-2.5) -- (6*\x+6,0);
\end{scope}
}

\end{tikzpicture}
\end{minipage}
\end{tabular}

\caption{One the left: the shape of the tiles in $\Gamma_{3,2}$. On the right: some levels in $\Gamma_{3,2}$. The level $\L$ has two successor levels (drawn below the level) and three predecessor levels (drawn above it).} \label{figure.basic_pattern_levels}

\end{figure}
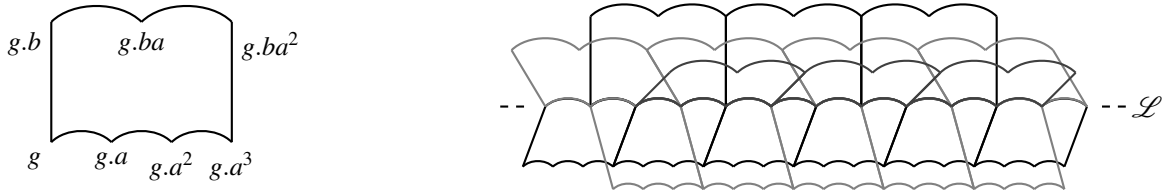

\subsection{Localization in $\Gamma_{m,n}$}

The Cayley graph $\Gamma_{m,n}$ can be projected into the Euclidean plane by a function $\Phi_{m,n}:\BS{m}{n}\rightarrow\R^2$, defined as follows. Let $w$ be a finite word on the alphabet $A=\{ a,b,a^{-1},b^{-1}\}$. Then any element of $\BS{m}{n}$ can be represented by such a word, but this representation is of course non unique. If $x$ is a letter of $A$, we denote by $|w|_x$ the number of occurrences of $x$ in the word $w$. We then define for $x\in A$ the contribution of $x$ to $w$ by $\parallel w\parallel_x=|w|_x-|w|_{x^{-1}}$.

Let $\psi_{m,n}:A^*\rightarrow\R$ be the function defined by induction on the length of the word by
$$
\left\{
\begin{array}{l}
\psi_{m,n}(\varepsilon)=0\text{ where }\varepsilon\text{ is the empty word}\\
\psi_{m,n}(w.b)=\psi_{m,n}(w.b^{-1})=\psi_{m,n}(w)\\
\psi_{m,n}(w.a)=\psi_{m,n}(w)+{\left(\frac{m}{n}\right)}^{\parallel
w\parallel_b}\\
\psi_{m,n}(w.a^{-1})=\psi_{m,n}(w)-{\left(\frac{m}{n}\right)}^{\parallel
w\parallel_b}
\end{array}
\right.
$$

\begin{lemma}\label{lemma.phi}
For every $u,v\in A^*$ one has
$$\psi_{m,n}(u.v)=\psi_{m,n}(u)+\left(\frac{m}{n}\right)^{\parallel u\parallel_{b}}\psi_{m,n}(v).$$
\end{lemma}

\begin{proof}
 By induction on the length of $v$.
\end{proof}

We can now define the projection function $\Phi_{m,n}:\BS{m}{n}\rightarrow \R^2$ which associates to every element $g$ of $\BS{m}{n}$ its coordinates on the Euclidean plane:
$$\Phi_{m,n}(g)=\left(\psi_{m,n}(w) ,\parallel w \parallel_{b^{-1}} \right),$$
where $w$ is a word representing $g$. The following proposition states that the definition does not depend on the choice of $w$. Its proof is a simple application of Lemma~\ref{lemma.phi}.


\begin{proposition}\label{proposition.Phi_well_defined}
 The function $\Phi_{m,n}$ is well defined on $\BS{m}{n}$.
\end{proposition}

\noindent
The function $\Phi_{m,n}$ is not necessarily injective as shown by the following example.

\begin{example}
\label{example.noninjective} Let $m=3$ and $n=2$. Consider the word 
$$\omega=bab^{-1}a^2ba^{-1}b^{-1}a^{-2}.$$
We have
$$\Phi_{3,2}(\omega)=\Phi_{3,2}(\varepsilon)=(0,0).$$
However,  freely reduced words that do not contain $b^{-1}a^{km}b$ or $ba^{kn}nb^{-1}$ as subwords, for any integer $k$, cannot
represent the identity in $\BS{m}{n}$. Thus $\omega$ and $\varepsilon$ represent different elements of the group. Moreover,
Baumslag-Solitar groups are HNN-extensions of $\Z$, thus from Britton's lemma it follows that a finite subgroup of
Baumslag-Solitar group is conjugate to a finite subgroup of $\Z$. Since $\omega$ is not the identity, it has infinite order in
$\BS{3}{2}$. We see that there is an infinite cyclic subgroup that is projected by $\Phi_{3,2}$ to point $(0,0)$.
\end{example}

All elements belonging to the same level project on the same horizontal line, thus we can speak of the \define{height of a level}. The height is $\parallel w \parallel_{b^{-1}}$ for the words $w$ that represent the elements of the level.

\section{Weakly aperiodic tile set on $\BS{m}{n}$}
\label{section.aperiodic}

The Cayley graph of $\BS{m}{n}$ consists in infinitely many hyperbolic layers that merge at various levels, and even if a strongly aperiodic tile set is known on the hyperbolic plane~\cite{DBLP:conf/mcu/Kari07}, it is not straightforward to synchronize all the layers to get an aperiodic tile set on $\BS{m}{n}$, and moreover we will see that strong aperiodicity is not even preserved. On the hyperbolic plane $\mathbb{H}^2$, there also exists a hierarchical strongly aperiodic tile set~\cite{DBLP:journals/tcs/Goodman-Strauss10} but it seems difficult, for the reason presented above, to adapt this construction to get a hierarchical weakly aperiodic tile set on~$\BS{m}{n}$.

\subsection{Representations of real numbers}
\label{subsection.representations}

Let $i\in\Z$. We say that a bi-infinite sequence $(x_k)_{k\in\Z}$ of $i$'s and $(i+1)$'s \define{represents} a real number $x\in [i,i+1]$ if there exists a sequence of intervals $I_1,I_2,\dots\subseteq\Z$ of increasing lengths $n_1<n_2<\dots$ such that
$$\lim_{k\rightarrow\infty} \frac{\sum_{j\in I_k} x_j}{n_k} =x,$$
that is to say there is an infinite sequence of intervals of increasing lengths whose averages converge to $x$. Note that if
$(x_k)_{k\in\Z}$ is a representation of $x$, all the shifted sequences $(x_{\ell+k})_{k\in\Z}$ for every $\ell\in\Z$ are also
representations of $x$. Note also that a sequence $(x_k)_{k\in\Z}$ can represent several distinct real numbers, since different
interval sequences may converge to different points, and that by a compactness argument, every sequence $(x_k)_{k\in\Z}$ does represent at least one real number $x$.

For a real number $x$, we denote by $\floor{x}$ the integer part of $x$, that is the largest integer which is lower or equal to $x$. Let $x,r\in \R$ be arbitrary. For every $k\in\Z$, let 
$$B_k=\floor{(r+k)x}-\floor{(r+k-1)x}.$$
The bi-infinite sequence $(B_k)_{k\in\Z}$ is called a \define{balanced representation of $x$}. Clearly $(B_k)_{k\in\Z}$ is a
representation of $x$ in the sense defined above. Balanced representations of irrational $x$ are sturmian sequences, while for rational $x$ the sequence is periodic.

\subsection{Dynamical systems as tilings}
\label{subsection.ds_tilings}

Let $I$ be a compact interval included in $\R$, and $T:I\rightarrow I$ be a dynamical system on $I$. We here restrict ourselves to piecewise affine maps. Following the ideas in~\cite{DBLP:journals/dm/Kari96} and \cite{DBLP:conf/mcu/Kari07} one can construct a tile set  such that if a level in a tiling of $\Gamma_{m,n}$  represents some real number $x$, then its successor levels represent $T(x)$. Going from one level to one of its successor level corresponds to one iteration of $T$, and a decrease of the height of the level by $1$.

\begin{figure}[H]
\begin{center}
\begin{tikzpicture}[scale=0.5]
 
\draw[thick] (0,0) -- (0,4);
\draw[thick] (0,4) to [controls=+(45:1) and +(135:1)] (3,4);
\draw[thick] (3,4) to [controls=+(45:1) and +(135:1)] (6,4);
\draw[thick] (6,4) -- (6,0);
\draw[thick] (0,0) to [controls=+(45:0.7) and +(135:0.7)] (2,0);
\draw[thick] (2,0) to [controls=+(45:0.7) and +(135:0.7)] (4,0);
\draw[thick] (4,0) to [controls=+(45:0.7) and +(135:0.7)] (6,0);


\draw (1.5,5) node{\small$x_1$};
\draw (4.5,5) node{\small$x_2$};
\draw (-1,2) node{\small$c$};
\draw (7,2) node{\small$d$};
\draw (1,1) node{\small$y_1$};
\draw (3,1) node{\small$y_2$};
\draw (5,1) node{\small$y_3$};

\end{tikzpicture}
\caption{The general form of tiles in $\BS{3}{2}$. The colors are numbers.} \label{figure.tilecolors}
\end{center}
\end{figure}
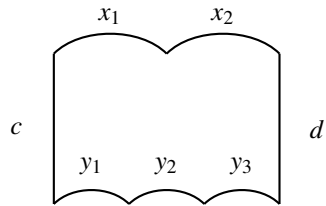

Consider the case $\BS{3}{2}$. The allowed patterns (tiles) are of the form shown in Figure~\ref{figure.tilecolors}. We say that the tiles multiply by number $q$ if the colors are numbers satisfying the relation
\begin{equation}
\label{equation1} q\frac{x_1+x_2}{2} + c = \frac{y_1+y_2+y_3}{3} +d.
\end{equation}

Consider a sequence of $k$ such tiles on some level, next to each other so that the left vertical edge of a tile is the same as the right edge of the previous tile. Averaging (\ref{equation1}) over all $k$ tiles yields then
$$
qx+\frac{c_1}{k}=y+\frac{d_k}{k}
$$
where $x$ is the average of the labels on the segment of $2k$ edges on the previous level, $y$ is the average of the labels on the corresponding segment of $3k$ edges on the level below, and $c_1$ and $d_k$ are the left and right vertical edges of the first and the last tile in the row, respectively. Letting $k$ grow to infinity, we see that if the previous level represents some number $x$ then the next level necessarily represents number $qx$, as required.



\subsection{An explicit tile set on $\BS{3}{2}$}
\label{subsection.explicit_tileset}

Let $T:[\frac{2}{3};2] \rightarrow [\frac{2}{3};2]$ be the piecewise linear map defined by
$$T: x \mapsto \left\{
  \begin{array}{c}
  2x\text{ if }x\in[\frac{2}{3};1]\\
  \frac{2}{3}x\text{ if }x\in]1;2]
  \end{array}
  \right.
$$
This function was used in~\cite{DBLP:journals/dm/Kari96} to construct an aperiodic set of Wang tiles on $\Z^2$. Here the same
idea yields a weakly aperiodic tile set on the Baumslag-Solitar group  $\BS{3}{2}$.






\begin{proposition}\label{proposition.T_aperiodic}
The dynamical system $T$ has no periodic points.
\end{proposition}

\begin{proof}
Suppose that $x\in[\frac{2}{3};2]$ is such that $T^n(x)=x$, with $n\geq 1$. Then there exist $k,\ell\in\N$ with $k+\ell=n$ and such that $2^k\left(\frac{2}{3}\right)^\ell x=x$. This is not possible since $x\neq0$ and $2$ and $3$ are mutually prime.
\end{proof}

We implement this dynamical system $T$ by a tile set $\tau$ on $\BS{3}{2}$. The tiles in $\tau$ have the form presented in
Figure~\ref{figure.tilecolors}. The tile set consists of two disjoint parts that implement multiplication by $2$ and
multiplication by $\frac{2}{3}$, respectively. The vertical edge identifies which part is used so that each level is tiled with tiles that multiply by the same number.

The tiles in the part that multiplies by $2$ have the following properties:
\begin{equation}
\label{condition1}
\left.
\begin{minipage}[l]{0.8\linewidth}
\begin{itemize}
\item the top colors $x_1,x_2\in\{0,1\}$,
\item the bottom colors $y_1,y_2,y_3\in\{1,2\}$, and
\item the relation (\ref{equation1}) is satisfied with $q=2$.
\end{itemize}
\end{minipage}
\right\}
\end{equation}
The tiles in the part that multiplies by $\frac{2}{3}$ satisfies the following:
\begin{equation}
\label{condition2}
\left.
\begin{minipage}[l]{0.8\linewidth}
\begin{itemize}
\item The top colors $x_1,x_2\in\{1,2\}$,
\item the bottom colors $y_1,y_2,y_3\in\{0,1\}$ or $y_1,y_2,y_3\in\{1,2\}$, and
\item the relation (\ref{equation1}) is satisfied with $q=\frac{2}{3}$.
\end{itemize}
\end{minipage}
\right\}
\end{equation}
Without specifying the tile set further, we can already see that it cannot admit a strongly periodic tiling. Based on the discussion above, if a level of a valid tiling represents a real number $x$ then its successor level necessarily represents
the number $T(x)$. (Or if it happens that $x=1$ then the successor might represent $\frac{2}{3}$ instead of $T(x)=2$, but
this causes no problems as the system still cannot have periodic points.)

\begin{proposition}\label{proposition.weakly_aperiodic}
There does not exist a strongly periodic valid tiling by $\tau$.
\end{proposition}

\begin{proof}
Assume that there is a strongly periodic tiling $\pi$. Because $\Per(\pi)$ has finite index $i$, every infinite subgroup of
$\BS{3}{2}$ contains a non-zero period. In particular, considering the subgroup $\{g a^k g^{-1}\ |\ k\in\Z\}$ for an arbitrary $g$, we see that there is a positive integer $k_g\leq i$ such that $g a^{k_g}g^{-1}\in \Per(\pi)$. Note that the translation by $g a^{k_g} g^{-1}$ translates the level $L_g$ by $k_g$ positions, that is, each $ga^j$ gets translated to $ga^{j+k_g}$. We see that $\pi$ is periodic on every level, and $p=i!$ is a common period.

Hence each level represents a unique number, and this number is a rational number of the form $\frac{c}{p}$ for some integer $c$. But there are only finitely many possible numbers since the labeling alphabets are finite. There exists hence $h>0$ such that $T^h(\frac{c}{p})=\frac{c}{p}$. Moreover, $c\neq 0$ because the choice of the alphabets in (\ref{condition1}) and (\ref{condition2}) prevent two consecutive levels to represent $0$. But this contradicts the fact that $T$ has no periodic points.
\end{proof}

\subsection{Weakly periodic valid tiling}

We next specify the tile set $\tau$ in more details to establish its weak aperiodicity. The tiles have to admit a valid tiling. To that purpose we may include in the tile set as many tiles as we need, as long as they satisfy the conditions (\ref{condition1}) and (\ref{condition2}) that guarantee that $\tau$ does not admit a strongly periodic tiling.

Our approach is to specify a coloring $\gamma$ of the edges of the Cayley graph and observe that at each vertex (\ref{condition1}) and (\ref{condition2}) are satisfied. We can then take the observed patterns in $\gamma$ as the set of allowed patterns. Our coloring $\gamma$ assigns the same color to all edges that are projected to the same position by the function $\Phi_{3,2}$.

Fix one bi-infinite orbit $(x_k)_{k\in\Z}$ by the dynamical system $T$ of the previous section. In other words, $x_{k}=T(x_{k-1})$ for all $k\in\Z$. Let $q_k\in\{\frac{2}{3},2\}$ be the used multiplier so that $x_{k}=q_kx_{k-1}$. We make in the following a corresponding coloring $\gamma$ of $\BS{3}{2}$. Let $g$ be an arbitrary group element projected to $\Phi_{3,2}(g)=\left(\alpha,\beta\right)$. Let $k=\beta$, and denote $x=x_k$ and $q=q_k$.
\begin{itemize}
\item The label of the edge $g\longrightarrow ga$ is
\begin{equation}
\label{equation.ga}
\bigfloor{\left(\frac{3}{2}\right)^{\beta}\left(\alpha+1\right) x}
-
\bigfloor{\left(\frac{3}{2}\right)^{\beta}\alpha x}.
\end{equation}
\item The label of the edge $g\longrightarrow gb$ is
\begin{equation}
\label{equation.gb}
\frac{q}{2}\bigfloor{\left(\frac{3}{2}\right)^{\beta-1}\alpha \frac{x}{q}}
-
\frac{1}{3}\bigfloor{\left(\frac{3}{2}\right)^{\beta}\alpha x}.
\end{equation}
Moreover, the label identifies $q=q_k$.
\end{itemize}

Note that the labels are determined by the coordinates $\alpha$ and $\beta$ so they are the same on all edges that project to the same point by $\Phi_{3,2}$. We have hence a well-defined edge-coloring of the projected Cayley graph. From (\ref{equation.ga})
we see that the edge $ga^j\longrightarrow ga^{j+1}$ is labeled by number
\begin{equation}
\label{equation.gaj}
\bigfloor{\left(\frac{3}{2}\right)^{\beta}\left(\alpha+j+1\right) x}
-
\bigfloor{\left(\frac{3}{2}\right)^{\beta}\left(\alpha+j\right) x}
\end{equation}
for all $j\in\Z$. Hence the sequence of labels on level $L_g$ is a balanced representation (defined in Section~\ref{subsection.representations}) of number $x=x_k$.

Consider the label $t$ of $g\longrightarrow gb$ given by (\ref{equation.gb}). For both $q=2$  and $q=\frac{2}{3}$ we have that
$t\in \frac{1}{3}\Z$, and by removing the floor functions in (\ref{equation.gb}) we easily get
$$
-\frac{1}{3}
<
t
<
\frac{q}{2}.
$$
We see that in the case $q=\frac{2}{3}$ the label is $0$, while in the case $q=2$ we have $t\in \{0,\frac{1}{3},\frac{2}{3}\}$.

Consider then the pattern of the form shown in Figure~\ref{figure.tilecolors} extracted from this coloring $\gamma$ at position $g$. More precisely, let the lower left corner of the shape be the group element $g$. The sum of the three bottom
edges $g\longrightarrow ga\longrightarrow ga^2\longrightarrow ga^3$ is the sum of (\ref{equation.gaj})
for $j=0,1,2$. The sum is telescoping:
$$
y_1+y_2+y_3 =
\bigfloor{\left(\frac{3}{2}\right)^{\beta}\left(\alpha+3\right) x}
-
\bigfloor{\left(\frac{3}{2}\right)^{\beta}\alpha x}.
$$
Analogously, the sum of the top edges $gb\longrightarrow gba \longrightarrow gba^2$ is
$$
x_1+x_2 =
\bigfloor{\left(\frac{3}{2}\right)^{\beta-1}\left(\alpha+2\right) \frac{x}{q}}
-
\bigfloor{\left(\frac{3}{2}\right)^{\beta-1}\alpha \frac{x}{q}}.
$$
From (\ref{equation.gb}) we get the colors of the edges $g\longrightarrow gb$ and
$ga^3\longrightarrow ga^3b$ as follows:
$$
\begin{array}{rcl}
c &=&
\frac{q}{2}\bigfloor{\left(\frac{3}{2}\right)^{\beta-1}\alpha \frac{x}{q}}
-
\frac{1}{3}\bigfloor{\left(\frac{3}{2}\right)^{\beta}\alpha x},\\ \\
d &=&
\frac{q}{2}\bigfloor{\left(\frac{3}{2}\right)^{\beta-1}\left(\alpha+2\right) \frac{x}{q}}
-
\frac{1}{3}\bigfloor{\left(\frac{3}{2}\right)^{\beta}\left(\alpha+3\right) x}.
\end{array}
$$
A direct calculation then shows that (\ref{equation1}) holds, that is, the pattern multiplies by $q$ as required. We conclude that the observed patterns in coloring $\gamma$ satisfy the conditions (\ref{condition1}) and (\ref{condition2}). Moreover, there are only finitely many different patterns since the labels come from finite sets. If we choose the observed patterns as the tiles, i.e. as the allowed patterns, we obtain a tile set that admits a valid tiling, but according to
Proposition~\ref{proposition.weakly_aperiodic} does not admit a strongly periodic tiling.

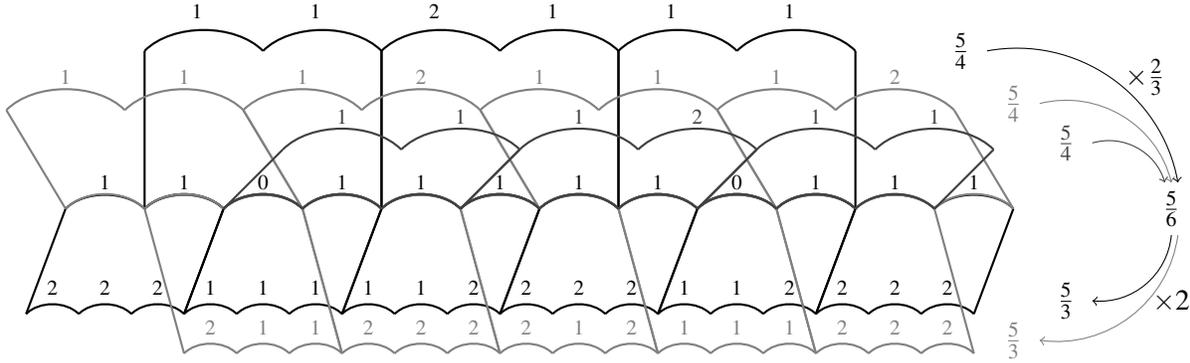
\begin{figure}[h]
\begin{bigcenter}
\begin{tikzpicture}[scale=0.35]

\draw (42,4) node{$\frac{5}{6}$};

\draw (34,10) node{$\frac{5}{4}$};
\draw (36,8) node[color=black!50]{$\frac{5}{4}$};
\draw (38,6.5) node[color=black!75]{$\frac{5}{4}$};

\draw[->] (35,10) to[bend left=45] (42.25,5);
\draw[->,color=black!50] (37,8) to[bend left=45] (42,5);
\draw[->,color=black!75] (39,6.5) to[bend left=45] (41.75,5);

\draw (38,0.5) node{$\frac{5}{3}$};
\draw (36,-1) node[color=black!50]{$\frac{5}{3}$};

\draw[->] (42,3) to[bend left=45] (39,0.5);
\draw[->,color=black!50] (42.25,3) to[bend left=45] (37,-1);

\draw (42,0.5) node{$\times 2$};
\draw (41,9) node{$\times \frac{2}{3}$};

\draw (-0.5,1) node{\scriptsize$2$};
\draw (1.5,1) node{\scriptsize$2$};
\draw (3.5,1) node{\scriptsize$2$};
\draw (5.5,1) node{\scriptsize$1$};
\draw (7.5,1) node{\scriptsize$1$};
\draw (9.5,1) node{\scriptsize$1$};
\draw (11.5,1) node{\scriptsize$1$};
\draw (13.5,1) node{\scriptsize$1$};
\draw (15.5,1) node{\scriptsize$2$};
\draw (17.5,1) node{\scriptsize$2$};
\draw (19.5,1) node{\scriptsize$2$};
\draw (21.5,1) node{\scriptsize$2$};
\draw (23.5,1) node{\scriptsize$1$};
\draw (25.5,1) node{\scriptsize$1$};
\draw (27.5,1) node{\scriptsize$2$};
\draw (29.5,1) node{\scriptsize$2$};
\draw (31.5,1) node{\scriptsize$2$};
\draw (33.5,1) node{\scriptsize$2$};

\draw (5.5,-0.5) node{\textcolor{black!50}{\scriptsize$2$}};
\draw (7.5,-0.5) node{\textcolor{black!50}{\scriptsize$1$}};
\draw (9.5,-0.5) node{\textcolor{black!50}{\scriptsize$1$}};
\draw (11.5,-0.5) node{\textcolor{black!50}{\scriptsize$2$}};
\draw (13.5,-0.5) node{\textcolor{black!50}{\scriptsize$2$}};
\draw (15.5,-0.5) node{\textcolor{black!50}{\scriptsize$2$}};
\draw (17.5,-0.5) node{\textcolor{black!50}{\scriptsize$2$}};
\draw (19.5,-0.5) node{\textcolor{black!50}{\scriptsize$1$}};
\draw (21.5,-0.5) node{\textcolor{black!50}{\scriptsize$2$}};
\draw (23.5,-0.5) node{\textcolor{black!50}{\scriptsize$1$}};
\draw (25.5,-0.5) node{\textcolor{black!50}{\scriptsize$1$}};
\draw (27.5,-0.5) node{\textcolor{black!50}{\scriptsize$1$}};
\draw (29.5,-0.5) node{\textcolor{black!50}{\scriptsize$2$}};
\draw (31.5,-0.5) node{\textcolor{black!50}{\scriptsize$2$}};
\draw (33.5,-0.5) node{\textcolor{black!50}{\scriptsize$2$}};

\draw (1.5,5) node{\scriptsize$1$};
\draw (4.5,5) node{\scriptsize$1$};
\draw (7.5,5) node{\scriptsize$0$};
\draw (10.5,5) node{\scriptsize$1$};
\draw (13.5,5) node{\scriptsize$1$};
\draw (16.5,5) node{\scriptsize$1$};
\draw (19.5,5) node{\scriptsize$1$};
\draw (22.5,5) node{\scriptsize$1$};
\draw (25.5,5) node{\scriptsize$0$};
\draw (28.5,5) node{\scriptsize$1$};
\draw (31.5,5) node{\scriptsize$1$};
\draw (34.5,5) node{\scriptsize$1$};

\draw (0,9) node{\textcolor{black!50}{\scriptsize$1$}};
\draw (4.5,9) node{\textcolor{black!50}{\scriptsize$1$}};
\draw (9,9) node{\textcolor{black!50}{\scriptsize$1$}};
\draw (13.5,9) node{\textcolor{black!50}{\scriptsize$2$}};
\draw (18,9) node{\textcolor{black!50}{\scriptsize$1$}};
\draw (22.5,9) node{\textcolor{black!50}{\scriptsize$1$}};
\draw (27,9) node{\textcolor{black!50}{\scriptsize$1$}};
\draw (31.5,9) node{\textcolor{black!50}{\scriptsize$2$}};

\draw (5,11.5) node{\scriptsize$1$};
\draw (9.5,11.5) node{\scriptsize$1$};
\draw (14,11.5) node{\scriptsize$2$};
\draw (18.5,11.5) node{\scriptsize$1$};
\draw (23,11.5) node{\scriptsize$1$};
\draw (27.5,11.5) node{\scriptsize$1$};

\draw (10.5,7.5) node{\textcolor{black!75}{\scriptsize$1$}};
\draw (15,7.5) node{\textcolor{black!75}{\scriptsize$1$}};
\draw (19.5,7.5) node{\textcolor{black!75}{\scriptsize$1$}};
\draw (24,7.5) node{\textcolor{black!75}{\scriptsize$2$}};
\draw (28.5,7.5) node{\textcolor{black!75}{\scriptsize$1$}};
\draw (33,7.5) node{\textcolor{black!75}{\scriptsize$1$}};

\foreach \x in {0,...,5} {
\draw[thick] (6*\x-1.5,0) -- (6*\x,4);
\draw[thick] (6*\x,4) to [controls=+(45:1) and +(135:1)] (6*\x+3,4);
\draw[thick] (6*\x+3,4) to [controls=+(45:1) and +(135:1)] (6*\x+6,4);
\draw[thick] (6*\x-1.5,0) to [controls=+(45:0.7) and +(135:0.7)] (6*\x+2-1.5,0);
\draw[thick] (6*\x+2-1.5,0) to [controls=+(45:0.7) and +(135:0.7)] (6*\x+4-1.5,0);
\draw[thick] (6*\x+4-1.5,0) to [controls=+(45:0.7) and +(135:0.7)] (6*\x+6-1.5,0);
\draw[thick] (6*\x+6,4) -- (6*\x+6-1.5,0);
}

\foreach \x in {0,...,4} {
\begin{scope}[shift={(3,0)},color=black!50]
\draw[thick] (6*\x+1.5,0-1.5) -- (6*\x,4);
\draw[thick] (6*\x,4) to [controls=+(45:1) and +(135:1)] (6*\x+3,4);
\draw[thick] (6*\x+3,4) to [controls=+(45:1) and +(135:1)] (6*\x+6,4);
\draw[thick] (6*\x+1.5,0-1.5) to [controls=+(45:0.7) and +(135:0.7)] (6*\x+2+1.5,0-1.5);
\draw[thick] (6*\x+2+1.5,0-1.5) to [controls=+(45:0.7) and +(135:0.7)] (6*\x+4+1.5,0-1.5);
\draw[thick] (6*\x+4+1.5,0-1.5) to [controls=+(45:0.7) and +(135:0.7)] (6*\x+6+1.5,0-1.5);
\draw[thick] (6*\x+6,4) -- (6*\x+6+1.5,0-1.5);
\end{scope}
}

\foreach \x in {0,1,2} {
\begin{scope}[scale=3/2,shift={(2,4*2/3)}]
\draw[thick] (6*\x,0) -- (6*\x,4);
\draw[thick] (6*\x,4) to [controls=+(45:1) and +(135:1)] (6*\x+3,4);
\draw[thick] (6*\x+3,4) to [controls=+(45:1) and +(135:1)] (6*\x+6,4);
\draw[thick] (6*\x,0) to [controls=+(45:0.7) and +(135:0.7)] (6*\x+2,0);
\draw[thick] (6*\x+2,0) to [controls=+(45:0.7) and +(135:0.7)] (6*\x+4,0);
\draw[thick] (6*\x+4,0) to [controls=+(45:0.7) and +(135:0.7)] (6*\x+6,0);
\draw[thick] (6*\x+6,4) -- (6*\x+6,0);
\end{scope}
}

\foreach \x in {0,1,2,3} {
\begin{scope}[scale=3/2,shift={(0,4*2/3)},color=black!50]
\draw[thick] (6*\x,0) -- (6*\x-1.5,4-1.5);
\draw[thick] (6*\x-1.5,4-1.5) to [controls=+(45:1) and +(135:1)] (6*\x+3-1.5,4-1.5);
\draw[thick] (6*\x+3-1.5,4-1.5) to [controls=+(45:1) and +(135:1)] (6*\x+6-1.5,4-1.5);
\draw[thick] (6*\x,0) to [controls=+(45:0.7) and +(135:0.7)] (6*\x+2,0);
\draw[thick] (6*\x+2,0) to [controls=+(45:0.7) and +(135:0.7)] (6*\x+4,0);
\draw[thick] (6*\x+4,0) to [controls=+(45:0.7) and +(135:0.7)] (6*\x+6,0);
\draw[thick] (6*\x+6-1.5,4-1.5) -- (6*\x+6,0);
\end{scope}
}

\foreach \x in {0,1,2} {
\begin{scope}[scale=3/2,shift={(4,4*2/3)},color=black!75]
\draw[thick] (6*\x,0) -- (6*\x+1.5,4-2.5);
\draw[thick] (6*\x+1.5,4-2.5) to [controls=+(45:1) and +(135:1)] (6*\x+3+1.5,4-2.5);
\draw[thick] (6*\x+3+1.5,4-2.5) to [controls=+(45:1) and +(135:1)] (6*\x+6+1.5,4-2.5);
\draw[thick] (6*\x,0) to [controls=+(45:0.7) and +(135:0.7)] (6*\x+2,0);
\draw[thick] (6*\x+2,0) to [controls=+(45:0.7) and +(135:0.7)] (6*\x+4,0);
\draw[thick] (6*\x+4,0) to [controls=+(45:0.7) and +(135:0.7)] (6*\x+6,0);
\draw[thick] (6*\x+6+1.5,4-2.5) -- (6*\x+6,0);
\end{scope}
}

\end{tikzpicture}
\caption{An example of tiling by $\tau$, corresponding to the orbit $\left(\dots,\frac{5}{4},\frac{5}{6},\frac{5}{3},\dots\right)$ in the dynamical system~$T$. For a better readability carries are omitted.}
\label{figure.example_orbit}
\end{bigcenter}
\end{figure}

\begin{proposition}\label{proposition.not_strongly_aperiodic}
There exists a weakly periodic valid tiling by $\tau$.
\end{proposition}

\begin{proof}
The coloring $\gamma$ described above is a valid tiling. According to Example~\ref{example.noninjective} the group element $h$ determined by the word
$$\omega=bab^{-1}a^2ba^{-1}b^{-1}a^{-2}$$
has infinite order and satisfies $\Phi_{3,2}(h)=(0,0)$. For any group element $g$ we have then $\Phi_{3,2}(hg)=\Phi_{3,2}(g)$ so that $h\in\Per(\gamma)$. So $\Per(\gamma)$ contains an infinite cyclic subgroup, and therefore $\gamma$ is weakly periodic.

\end{proof}

We can explicitly construct $\tau$ by taking all tiles that satisfy (\ref{condition1}) and whose vertical labels $c$ and $d$ come from the set $\{0,\frac{1}{3},\frac{2}{3}\}$, and all tiles that satisfy (\ref{condition2}) and where $c=d=0$. All shapes that appear in $\gamma$ are then included. It turns out the first part (that multiplies by $q=2$) contains 36 tiles, and the second part (for $q=\frac{2}{3}$) contains 10 tiles. These are depicted in Figure~\ref{figure.strongly_aperiodic_tiles_set}.

\begin{figure}[h]
\begin{bigcenter}
\begin{tikzpicture}[scale=0.3]


\foreach \x in {1,3} {
\begin{scope}[shift={(\x*7,24)},scale=0.8]
\draw[thick,fill=black!20] (0,0) -- (0,4)-- (0,4) to [controls=+(45:1) and +(135:1)] (3,4) -- (3,4) to [controls=+(45:1) and +(135:1)] (6,4) -- (6,4) -- (6,0) -- (6,0) to [controls=+(135:0.7) and +(45:0.7)] (4,0) -- (4,0) to [controls=+(135:0.7) and +(45:0.7)] (2,0) -- (2,0) to [controls=+(135:0.7) and +(45:0.7)] (0,0) -- cycle ;
\end{scope}
}

\begin{scope}[shift={(21,24)},scale=0.8]
\draw (1.5,5) node{\scriptsize$1$};
\draw (4.5,5) node{\scriptsize$2$};
\draw (-1,2) node{\scriptsize$0$};
\draw (7,2) node{\scriptsize$0$};
\draw (3,-1) node{\scriptsize\textbf{(2 tiles)}};
\draw (1,1) node{\scriptsize$1$};
\draw (3,1) node{\scriptsize$1$};
\draw (5,1) node{\scriptsize$1$};
\end{scope}

\begin{scope}[shift={(7,24)},scale=0.8]
\draw (1.5,5) node{\scriptsize$2$};
\draw (4.5,5) node{\scriptsize$2$};
\draw (-1,2) node{\scriptsize$0$};
\draw (7,2) node{\scriptsize$0$};
\draw (3,-1) node{\scriptsize\textbf{(3 tiles)}};
\draw (1,1) node{\scriptsize$1$};
\draw (3,1) node{\scriptsize$1$};
\draw (5,1) node{\scriptsize$2$};
\end{scope}


\foreach \x in {1,3} {
\begin{scope}[shift={(\x*7,32)},scale=0.8]
\draw[thick,pattern color=black!20 ,pattern=north east lines] (0,0) -- (0,4)-- (0,4) to [controls=+(45:1) and +(135:1)] (3,4) -- (3,4) to [controls=+(45:1) and +(135:1)] (6,4) -- (6,4) -- (6,0) -- (6,0) to [controls=+(135:0.7) and +(45:0.7)] (4,0) -- (4,0) to [controls=+(135:0.7) and +(45:0.7)] (2,0) -- (2,0) to [controls=+(135:0.7) and +(45:0.7)] (0,0) -- cycle ;
\end{scope}
}

\begin{scope}[shift={(21,32)},scale=0.8]
\draw (1.5,5) node{\scriptsize$1$};
\draw (4.5,5) node{\scriptsize$2$};
\draw (-1,2) node{\scriptsize$0$};
\draw (7,2) node{\scriptsize$0$};
\draw (3,-1) node{\scriptsize\textbf{(2 tiles)}};
\draw (1,1) node{\scriptsize$1$};
\draw (3,1) node{\scriptsize$1$};
\draw (5,1) node{\scriptsize$1$};
\end{scope}

\begin{scope}[shift={(7,32)},scale=0.8]
\draw (1.5,5) node{\scriptsize$1$};
\draw (4.5,5) node{\scriptsize$1$};
\draw (-1,2) node{\scriptsize$0$};
\draw (7,2) node{\scriptsize$0$};
\draw (3,-1) node{\scriptsize\textbf{(3 tiles)}};
\draw (1,1) node{\scriptsize$0$};
\draw (3,1) node{\scriptsize$1$};
\draw (5,1) node{\scriptsize$1$};
\end{scope}


\foreach \x in {0,2,4} {
\begin{scope}[shift={(\x*7,40)},scale=0.8]
\draw[thick] (0,0) -- (0,4)-- (0,4) to [controls=+(45:1) and +(135:1)] (3,4) -- (3,4) to [controls=+(45:1) and +(135:1)] (6,4) -- (6,4) -- (6,0) -- (6,0) to [controls=+(135:0.7) and +(45:0.7)] (4,0) -- (4,0) to [controls=+(135:0.7) and +(45:0.7)] (2,0) -- (2,0) to [controls=+(135:0.7) and +(45:0.7)] (0,0) -- cycle ;
\end{scope}
}

\begin{scope}[shift={(0,40)},scale=0.8]
\draw (1.5,5) node{\scriptsize$1$};
\draw (4.5,5) node{\scriptsize$1$};
\draw (-1,2) node{\scriptsize$0$};
\draw (7,2) node{\scriptsize$\frac{2}{3}$};
\draw (3,-1) node{\scriptsize\textbf{(3 tiles)}};
\draw (1,1) node{\scriptsize$1$};
\draw (3,1) node{\scriptsize$1$};
\draw (5,1) node{\scriptsize$2$};
\end{scope}

\begin{scope}[shift={(14,40)},scale=0.8]
\draw (1.5,5) node{\scriptsize$1$};
\draw (4.5,5) node{\scriptsize$1$};
\draw (-1,2) node{\scriptsize$c$};
\draw (8,2) node{\scriptsize$c+\frac{1}{3}$};
\draw (3,-1) node{\scriptsize$c\in\{0,\frac{1}{3}\}$};
\draw (3,-2.5) node{\scriptsize\textbf{(6 tiles)}};
\draw (1,1) node{\scriptsize$1$};
\draw (3,1) node{\scriptsize$2$};
\draw (5,1) node{\scriptsize$2$};
\end{scope}

\begin{scope}[shift={(28,40)},scale=0.8]
\draw (1.5,5) node{\scriptsize$1$};
\draw (4.5,5) node{\scriptsize$1$};
\draw (-1,2) node{\scriptsize$c$};
\draw (7,2) node{\scriptsize$c$};
\draw (3,-1) node{\scriptsize$c\in\{0,\frac{1}{3},\frac{2}{3}\}$};
\draw (3,-2.5) node{\scriptsize\textbf{(3 tiles)}};
\draw (1,1) node{\scriptsize$2$};
\draw (3,1) node{\scriptsize$2$};
\draw (5,1) node{\scriptsize$2$};
\end{scope}

\foreach \x in {0,2,4} {
\begin{scope}[shift={(\x*7,48)},scale=0.8]
\draw[thick] (0,0) -- (0,4)-- (0,4) to [controls=+(45:1) and +(135:1)] (3,4) -- (3,4) to [controls=+(45:1) and +(135:1)] (6,4) -- (6,4) -- (6,0) -- (6,0) to [controls=+(135:0.7) and +(45:0.7)] (4,0) -- (4,0) to [controls=+(135:0.7) and +(45:0.7)] (2,0) -- (2,0) to [controls=+(135:0.7) and +(45:0.7)] (0,0) -- cycle ;
\end{scope}
}

\begin{scope}[shift={(0,48)},scale=0.8]
\draw (1.5,5) node{\scriptsize$0$};
\draw (4.5,5) node{\scriptsize$1$};
\draw (-1,2) node{\scriptsize$c$};
\draw (7,2) node{\scriptsize$c$};
\draw (3,-1) node{\scriptsize$c\in\{0,\frac{1}{3},\frac{2}{3}\}$};
\draw (3,-2.5) node{\scriptsize\textbf{(6 tiles)}};
\draw (1,1) node{\scriptsize$1$};
\draw (3,1) node{\scriptsize$1$};
\draw (5,1) node{\scriptsize$1$};
\end{scope}

\begin{scope}[shift={(14,48)},scale=0.8]
\draw (1.5,5) node{\scriptsize$0$};
\draw (4.5,5) node{\scriptsize$1$};
\draw (-1,2) node{\scriptsize$c$};
\draw (8,2) node{\scriptsize$c-\frac{1}{3}$};
\draw (3,-1) node{\scriptsize$c\in\{\frac{1}{3},\frac{2}{3}\}$};
\draw (3,-2.5) node{\scriptsize\textbf{(12 tiles)}};
\draw (1,1) node{\scriptsize$1$};
\draw (3,1) node{\scriptsize$1$};
\draw (5,1) node{\scriptsize$2$};
\end{scope}

\begin{scope}[shift={(28,48)},scale=0.8]
\draw (1.5,5) node{\scriptsize$0$};
\draw (4.5,5) node{\scriptsize$1$};
\draw (-1,2) node{\scriptsize$\frac{2}{3}$};
\draw (7,2) node{\scriptsize$0$};
\draw (3,-1) node{\scriptsize\textbf{(6 tiles)}};
\draw (1,1) node{\scriptsize$1$};
\draw (3,1) node{\scriptsize$2$};
\draw (5,1) node{\scriptsize$2$};
\end{scope}

\end{tikzpicture}

\caption{The tile set $\tau$, defined on the group $\BS{3}{2}$, that implements the dynamical system $T$. The whole set of 46 tiles can be easily obtained by permutations of the values for $x_1$, $x_2$ and for $y_1$,$y_2$, $y_3$ and by choosing appropriate values for the carries $c$. The tiles corresponding to multiplication by $2$ are pictured in white, those to multiplication by $\frac{2}{3}$ from a real number $x\in]1;\frac{3}{2}]$ in dashed grey, and from a real number $x\in ]\frac{3}{2};2]$ in grey.}
\label{figure.strongly_aperiodic_tiles_set}
\end{bigcenter}
\end{figure}
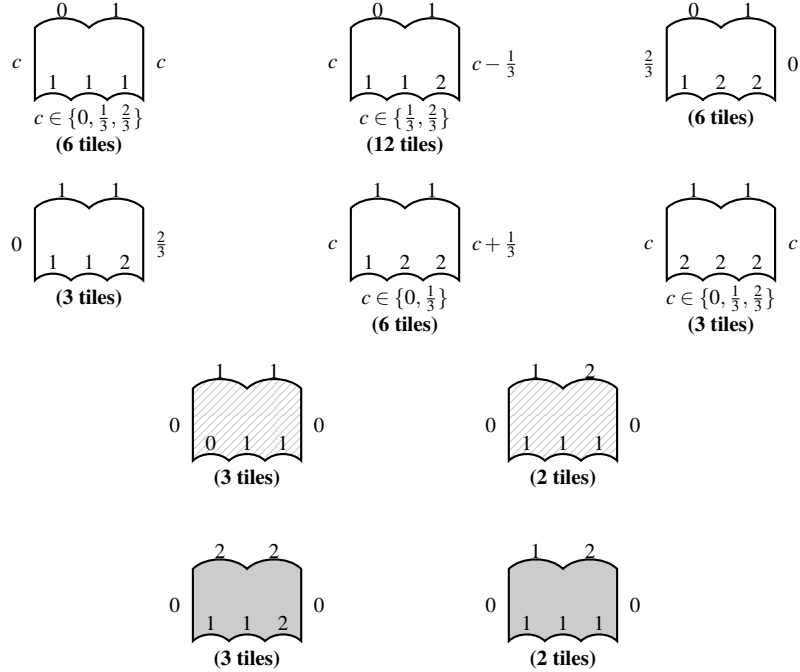

A similar tile set can be constructed on $\BS{m}{n}$ for every $m,n>0$. The dynamical system $T$ presented in Section~\ref{subsection.explicit_tileset} can be used for all the groups $\BS{m}{n}$ -- we have chose to present an explicit tile set on $\BS{3}{2}$ only because it simplifies pictures and gives a reasonable size for $\tau$, and it has no connection with the mutually primeness of $2$ and $3$ which makes the dynamical system $T$ aperiodic -- and arguments for proving the weak aperiodicity of the tile set remain the same.

\begin{theorem}
 There exist weakly aperiodic tile sets on $\BS{m}{n}$ for every $m,n>0$.
\end{theorem}

\section{Undecidability of the domino problem on $\BS{m}{n}$}
\label{section.dp_undecidable}

We adapt the proof of undecidability of the domino problem in the hyperbolic plane presented in~\cite{DBLP:conf/mcu/Kari07}.
The proof is by reduction of the mortality problem of piecewise affine maps, which is known to be undecidable.


Given a finite number of rational affine transformations $f_1,f_2,\dots,f_n$ of $\R^2$ and disjoint closed unit squares $U_1,U_2,\dots,U_n$ that serves as domains for the corresponding affine map, we define a partial function $f:\R^2\rightarrow\R^2$ whose domain is $U=U_1\cup U_2 \cup\dots\cup U_n$ and which is defined by
$$\overrightarrow{x}\mapsto f_i(\overrightarrow{x})\text{ for
}\overrightarrow{x}\in U_i.$$
A vector $\overrightarrow{x}\in\R^2$ is \define{immortal} if for every $i\in\Z$, the value $f^i(\overrightarrow{x})$ is always in the domain $U$. The \define{mortality problem of piecewise affine maps} asks whether a given system of rational affine transformations of the plane $f_1,f_2,\dots,f_n$ and disjoint unit squares $U_1,U_2,\dots,U_n$ with integer corners has an immortal point. It is known, using the undecidability of the Turing machine mortality problem~\cite{Hooper1966}, that the mortality problem of piecewise affine maps is undecidable~\cite{DBLP:conf/mcu/Kari07}.


Let $f_i:\R^2\rightarrow\R^2$ be a rational affine transformation of the Euclidean plane with domain $U_i$, given by
$$f_i(\overrightarrow{x})=M_i\overrightarrow{x}+\overrightarrow{b_i}.$$
Following~\cite{DBLP:conf/mcu/Kari07}, one can adapt the tiles presented in Section~\ref{subsection.ds_tilings} to construct a finite tile set that computes function $f_i$ and whose top edges labels represent a vector in $U_i$.


Given a piecewise affine map $f$ with domain $U$, one can thus effectively construct a finite tile set $\tau$ on $\BS{m}{n}$ such that $\tau$ admits a valid tiling if and only if $f$ has an immortal point. From the undecidability of the mortality
problem for piecewise affine maps~\cite{DBLP:conf/mcu/Kari07} we deduce

\begin{theorem}
 The domino problem is undecidable on Baumslag-Solitar groups.
\end{theorem}

As explained in~\cite{DBLP:conf/mcu/Kari07}, it is possible to encode Turing machines into affine maps, and affine maps can be
encoded in a tiling on Baumslag-Solitar groups as discussed above. Combining this two facts, we can conclude that for any Baumslag-Solitar group, it is possible to encode any Turing machine inside a tiling on this group. It was shown in~\cite{Jeandel2012} that there exists a Turing machine $M$ that halts on every recursive configuration but does not halt on
some non-recursive configuration. Executing our construction on $M$ one obtains a tile set on $\BS{m}{n}$ that admits a valid tiling, and all valid tilings are non recursive.

\begin{theorem}
 There exist an arecursive tile set on $\BS{m}{n}$ for every $m,n>0$.
\end{theorem}



\section*{Acknowledgement}

This work has been supported by the Academy of Finland Grant 131558 and by ANR project SubTile.

\nocite{*}
\bibliographystyle{eptcs}
\bibliography{AubrunKari}

\end{document}